%% file: root.tex
\documentclass[letterpaper, 10 pt, conference]{ieeeconf}  

\IEEEoverridecommandlockouts                              
\title{\LARGE \bf
Robust Optimal Control for Nonlinear Systems with Parametric Uncertainties via System Level Synthesis
}
\author{Antoine P. Leeman$^1$ \and Jérôme Sieber$^1$ \and Samir Bennani$^2$ \and Melanie N. Zeilinger$^1$
\thanks{This work has been supported by the European Space Agency under OSIP 4000133352, under Grant NPI 621-2018 and by the Swiss Space Center. Corresponding author: Antoine P. Leeman. Email: aleeman@ethz.ch}
\thanks{$^1$Antoine P. Leeman, Jérôme Sieber, and Melanie N. Zeilinger are with the IDSC, ETH Zurich, 8092 Zurich, Switzerland}
\thanks{$^2$Samir Bennani is with the ESTEC, European Space Agency, 2201 AZ Noordwijk, The Netherlands}}
\usepackage{array}
\usepackage{cite}
\usepackage{amssymb,amsfonts}
\usepackage{algorithmic}
\usepackage{graphicx}
\usepackage{textcomp}
\usepackage{xcolor}
\usepackage{graphics}
\usepackage{amssymb}
\usepackage{xcolor}
\usepackage{bm}
\usepackage{graphicx}
\usepackage[intlimits]{amsmath}
\usepackage{mathtools}
\usepackage{comment}
\usepackage{printlen}
\usepackage{placeins}
\usepackage[capitalise]{acro}
\usepackage{ifthen}
  \usepackage{soul}
\usepackage{booktabs}
\usepackage{balance}

\usepackage{tikz, adjustbox}
\usepackage[most]{tcolorbox}
\usepackage{xcolor}
\usepackage{wrapfig}

\input{notes}
\usepackage{hyperref}

\DeclareAcronym{RMPC}{
  short = RMPC,
  long  = robust model predictive control,
}
\DeclareAcronym{SLS}{
  short = SLS,
  long  = system level synthesis,
}
\DeclareAcronym{MPC}{
  short = MPC,
  long  = model predictive control,
}

\DeclareAcronym{LTV}{
  short = LTV,
  long  = linear time-varying,
}
\DeclareAcronym{NLP}{
  short = NLP,
  long  = nonlinear program,
}

\newcommand{\xz}{\bar x}
\newcommand{\Z}{\mathbf{z}}
\newcommand{\V}{\mathbf{v}}
\newcommand{\U}{\mathbf{u}}
\renewcommand{\d}{{\tilde{\mathbf{d}}}}
\renewcommand{\P}{\mathbf \Phi}
\newcommand{\Px}{\P_\x}
\newcommand{\Pu}{\P_\u}

\newcommand{\Y}{\mathbf{y}}
\newcommand{\X}{\mathbf{x}}

\newcommand{\nx}{{n_\mathrm{x}}}
\renewcommand{\nu}{{n_\mathrm{u}}}
\newcommand{\nt}{{n_\mathrm{\theta}}}
\newcommand{\nw}{{n_\mathrm{w}}}
\newcommand{\nwt}{{n_{\tilde{\mathrm{w}}}}}
\newcommand{\ny}{{n_\mathrm{y}}}
\newcommand{\R}{\mathbb{R}}
\newcommand{\vertx}[1]{{#1}_\textrm{vert}}

\newcommand{\W}{\mathbf{w}}
\newcommand{\E}{\bm{\eta}}
\newcommand{\fl}{\bar f}
\newcommand{\fa}{f_\theta}

\newcommand{\sv}[2]{\left(#1,#2\right)}
\newcommand{\T}{^\top}
\newcommand{\Km}{\forall k\in \mathbb{N}_{T}}
\newcommand{\K}{\forall k\in \mathbb{N}_{T+1}}

\newcommand{\I}{\forall i\in \mathbb{N}_{n_c}}

\newcommand{\x}{\mathrm{x}}
\renewcommand{\u}{\mathrm{u}}

\newcommand{\e}{h}
\newcommand{\tube}{\bm \tau}

\newcommand{\B}{ \mathcal{B}_\infty}
\newcommand{\defmath}{\vcentcolon =}
\newcommand{\mathdef}{=\vcentcolon}

\newcommand{\blue}[1]{{\color{blue}{#1}}}
\newcommand{\onenorm}[1]{ \|{\e_i^\top} #1 \|_1 }

\newtheorem{remark}{Remark}[section]

\newtheorem{proposition}{Proposition}[section]
\newtheorem{theorem}{Theorem}[section]
\newtheorem{assumption}{Assumption}[section]
\newtheorem{corollary}{Corollary}[section]
\begin{document}

\maketitle
\thispagestyle{empty}
\pagestyle{empty}

\begin{abstract}
This paper addresses the problem of optimally controlling nonlinear systems with norm-bounded disturbances and parametric uncertainties while robustly satisfying constraints. The proposed approach jointly optimizes a nominal nonlinear trajectory and an error feedback, requiring minimal offline design effort and offering low conservatism.
This is achieved by decomposing the affine-in-the-parameter uncertain nonlinear system into a nominal \emph{nonlinear} system and an uncertain linear time-varying system.
Using this decomposition, we can apply established tools from system level synthesis to \emph{convexly} over-bound all uncertainties in the nonlinear optimization problem. Moreover, it enables tight joint optimization of the linearization error bounds, parametric uncertainties bounds, nonlinear trajectory, and error feedback.
With this novel controller parameterization, we can formulate a convex constraint to ensure robust performance guarantees for the nonlinear system.
The presented method is relevant for numerous applications related to trajectory optimization, e.g., in robotics and aerospace engineering. We demonstrate the performance of the approach and its low conservatism through the simulation example of a post-capture satellite stabilization.
\end{abstract}
\section{Introduction}
\subsection{Motivation}
Robust nonlinear optimal control addresses a ubiquitous challenge in various safety-critical applications, such as drones, spacecraft, and robotic systems~\cite{Ho2021OnlineUncertainty,Wu2022Robust-RRT:Systems, Singh2017,leeman2023autonomous}. These systems are complex and difficult to model accurately due to uncertainties from  measurement errors, unmodeled dynamics, or environmental disturbances. 
As a result, there is often a mismatch between the predictive model and the actual system.
It is therefore crucial to guarantee robust constraint satisfaction to ensure the safety of these systems.
However, achieving such guarantees often comes at the cost of conservatism, non-trivial system-specific design, or substantial computational effort (either during operation or during controller design), originating from trading-off performance, robustness, and flexibility, as evidenced by a wealth of prior research on the topic outlined below.

\subsection{Related Work}
Traditionally, robust optimal control has been divided into two main steps: (1) the optimization of a nominal trajectory (also called reference trajectory, guidance, or feed-forward)~\cite{MALYUTA2021282} and (2) the offline design of a stabilizing feedback (controller) to compensate for modeling errors or disturbances. The propagation of those uncertainties is commonly approximated using tubes, or funnels~\cite{majumdar2017funnel}, and the corresponding reachable sets around the nominal trajectory are used to ensure robust constraint satisfaction.
However, this separation may introduce significant conservatism and limit the controller's performance; compare also~\cite{villanueva2017robust} for a recent approach for nonlinear system.
In this work, we propose an approach that addresses the limitations of this classic design paradigm by considering joint optimization of the nominal trajectory and error feedback, with a focus on robust constraint satisfaction in nonlinear systems with disturbances.

Specifically, this work aims to tackle the challenges associated with a particular class of {uncertainties}: parametric uncertainties, which commonly arise from model mismatch.
Techniques for addressing such uncertainties have been developed in the context of \ac{RMPC}.
Here, a trajectory optimization problem is solved at each time step, incorporating a fixed feedback and tightened constraints.
The conservatism of the constraint tightening is well-studied for linear systems~\cite{Kohler2019LinearConservatism,fleming2014multiplicative,Bujarbaruah2021AUncertainty}; however, extending it to nonlinear systems with parametric uncertainties remains a significant challenge in general. Although several computationally efficient methods have been proposed, using non-trivial offline designs based on contraction metrics~\cite{Sasfi2022RobustMetrics, Kohler2021ASystems, kohler2021robust}, these approaches often result in system-dependent or conservative designs. 

With the aim of reducing conservatism, \ac{SLS}~\cite{Anderson2019} has been introduced to jointly and convexly optimize the error feedback and nominal trajectory. 
Because the error feedback is optimized online, the offline system-specific design is reduced, and \ac{SLS}-based \ac{RMPC} typically exhibits reduced conservatism~\cite{leeman2022predictive,sieber2021system}, especially with parametric uncertainties~\cite{Chen2022RobustSynthesis}.
While \ac{SLS} has been extended to nonlinear systems~\cite{Dimitar2020, conger2022nonlinear,furieri2022neural,chen2023safety}, including a formulation with robust constraint satisfaction~\cite{Leeman2023RobustSynthesis}, none of the previous approaches have considered parametric uncertainties. In this paper, we fill this gap by treating parametric uncertainties similarly as in linear \ac{SLS}~\cite{Chen2022RobustSynthesis}, all while optimizing the nonlinear trajectory.%
\subsection{Contribution}
We present a novel approach for jointly optimizing an error feedback and nominal trajectory for a nonlinear system with affine parametric uncertainties while ensuring robust constraint satisfaction.
Besides the controller and the emph{nonlinear} trajectory, a \emph{convex} overbound of the parametric uncertainties and linearization errors is also jointly optimized, leveraging \ac{SLS}. This presents an advantage over~\cite{Leeman2023RobustSynthesis}, where the over-bound is neither convex nor accounts for {affine} parametric uncertainties.
The convex over-bounding enables features such as convex constraints to guarantee robust performance.
This new controller parameterization significantly reduces the need for offline, system-specific design, typically required in conventional methods that rely on offline-designed stabilizing controllers. Consequently, our approach is less conservative.
The benefits are demonstrated through the simulation of the post-capture satellite stabilization with uncertainties.
By combining low conservatism and flexibility, our proposed method holds significant potential for new developments in nonlinear \ac{RMPC} and motion planning.
\subsection{Notation}
We define the set $\mathbb{N}_T \defmath \{0,\dots,T-1\}$ where $T$ is a natural number. We denote stacked vectors or matrices by $\sv{a}{b} = [a\T,~b\T]\T$. For a vector $r\in \R^n$, we denote its $i^\text{th}$ component by $r_i$. Let $\R$ be the set of real numbers, and $0_{p,q}\in \R^{p,q}$ be a matrix of zeros. Let $\mathcal{L}^{T,p\times q}$ denote the set of all block lower-triangular matrices with the following structure
\begin{equation}
     {M} = \begin{bmatrix} M^{0,0} & 0_{p,q} & \dots & 0_{p,q} \\ M^{1,1} & M^{1,0} & \dots & 0_{p,q} \\ \vdots & \vdots & \ddots & \vdots \\M^{T-1,T-1} & M^{T-1,T-2} & \dots & M^{T-1,0} \end{bmatrix},
     \label{eq:mat}
\end{equation}
where $M^{i,j}\in \R^{p\times q}$. We denote the $k^\text{th}$ block row of $M$ as {$M^{(k)}\defmath [M^{k,k}, \ldots,~M^{k,0},0_{p,q(T-1-k)}]$.}
The block diagonal matrix consisting of matrices $A_1,\dots,A_T$ is defined by $\mathrm{blkdiag}(A_1,\dots,A_T)$.
The matrix ${I}$ denotes the identity with its dimensions either inferred from the context or indicated by the subscript, i.e., ${I}_{\nx}\in\R^{\nx\times \nx}$.
Let $\B^{m}$ be the unit ball defined by $\B^{m} \defmath \{d\in\R^m|~ \|d\|_\infty \le 1\}$.
We define $\mathcal{W}^k\defmath \underbrace{\mathcal{W} \times \dots \times \mathcal{W}}_{\text{$k$ times}}$.
For a sequence of vectors $\W_k\in\mathcal{W}_k\subseteq \R^m$ and $k\in\mathbb{N}$, we define $\W_{0:k}:=(\W_0,\dots,\W_k)\in \mathcal{W}^{0:k}\defmath\mathcal{W}_0 \times \dots \times \mathcal{W}_{k}$. For a set $\Theta$, we write the subset consisting of its vertices as $\Theta_\textrm{vert} \subset \Theta$.

\section{Problem Formulation}

Consider the robust nonlinear optimal control problem characterized by 
\begin{subequations}
\label{eq:prob_form}
\begin{align}
\min_{\pi(.)} \quad & J_T(\xz, \pi(.)),\label{eq:cost_horizon}\\
    \text{s.t.}\quad  & \Km : \nonumber \\
    &\X_{k+1}= f(\X_{k},\U_{k}, \theta) + \W_k,~ \X_0 = \bar x,\label{eq:nonlinear_dyn}\\
    & \U_k = \pi_k(\X_{0:k}),\label{eq:robust_controller}\\
    &{(\X_k,\U_k) \in \mathcal{C},~\forall \W_k \in \mathcal{W}~\forall \theta \in \Theta.}\label{eq:prob_form_cons}
\end{align}
\end{subequations}
The dynamics are given by \eqref{eq:nonlinear_dyn}, with state $\X_k \in \R^\nx$, input $\U_k \in \R^\nu$, parameter $\theta\in\R^\nt$, and norm-bounded disturbance 
\begin{equation}
    \W_k\in \mathcal{W}\defmath E \B^\nw,\label{eq:disturbance_E}
\end{equation}
with $E\in \R^{\nx \times \nw}$. 
The results of this paper can be extended to more general polytopic sets $\mathcal{W}\defmath \{ w \in \R^\nw|~ H_w w \le h_w\}$ although this extension was omitted for notational clarity.
The fixed but unknown parameter $\theta$ belongs to a known compact parameter set $\Theta \subset \R^{n_\theta}$.
The initial condition is given by $\bar x \in \R^\nx$. 
The feedback~\eqref{eq:robust_controller} is designed to minimize the objective~\eqref{eq:cost_horizon}, while, robustly satisfying the constraints~\eqref{eq:prob_form_cons}, i.e., for any realization of the disturbance $\W_k$ and any value of the parameter $\theta\in\Theta$.
\begin{assumption}
The dynamics~\eqref{eq:nonlinear_dyn} are affine in~$\theta$, i.e.,
\begin{equation}
\begin{aligned}
   \X_{k+1} & = \fl(\X_k, \U_k) +  f_\theta(\X_k, \U_k) \theta+ \W_k,\\
    & = \fl(\X_k, \U_k) +  \sum_{i=1}^\nt f_{\theta,i}(\X_k, \U_k) \theta_i+ \W_k,
\end{aligned}
\label{eq:dynamics}
\end{equation}
with a nominal model $\fl$ and a parametric model $f_\theta$.
\label{assum:0}
\end{assumption}
To ensure computational tractability, we select a causal, affine, time-varying error feedback
\begin{equation}
\label{eq:affine_feedback}
    \pi_k(\X_{0:k}) \defmath \V_k + \sum_{j=0}^{k-1} K^{k-1,j} \Delta \X_{k-j},
\end{equation}
based on a nominal trajectory
\begin{equation}
    \Z_{k+1} = \fl(\Z_k, \V_k),~\Z_0 = \bar x.
    \label{eq:ref_dynamics}
\end{equation} 
The error between the state governed by~\eqref{eq:dynamics} and the nominal trajectory~\eqref{eq:ref_dynamics} is denoted by $\Delta \X_k  \defmath  \X_k  - \Z_k $, with a similar definition for the input error $\Delta \U_k  \defmath \U_k - \V_k$. 
\begin{assumption}
The nonlinear functions $\fl: \R^{\nx} \times \R^{\nu} \mapsto \R^{\nx}$ and $f_\theta:  \R^{\nx} \times \R^{\nu} \mapsto \R^\nx \times \R^\nt$ are three times continuously differentiable.
\label{assum:1}
\end{assumption}
\begin{assumption}
The constraint set $\mathcal{C}$ is a compact polytopic set defined as
\begin{equation}
\begin{aligned}
    \mathcal{C}& \defmath  \{(x,u)|~ c_i\T(x,u)+b_i\leq 0, \forall i\in \mathbb{N}_{n_c}\},\label{eq:constraints}\\
    \end{aligned}
\end{equation}where $c_i \in \R^{\nx+\nu}$, and $b_i \in \R$.
\label{assum:2}
\end{assumption}
\begin{assumption}
The set of parameters $\Theta$ is a compact polytopic set $ \Theta \defmath \{ \theta \in \R^\nt|~ H_\theta \theta \le h_\theta \}$.
\label{assum:3}
\end{assumption}

\section{Robust nonlinear optimal control via \ac{SLS}}
In this section, we present the main contribution of the paper, which is a tractable formulation of~\eqref{eq:prob_form} for optimizing jointly an affine policy~\eqref{eq:affine_feedback} and a nonlinear trajectory that ensure robust constraint satisfaction for the uncertain nonlinear system~\eqref{eq:dynamics}.
The proposed approach relies on linearization, discussed in Section~\ref{sec:linearization}.
It decomposes the uncertain nonlinear system~\eqref{eq:dynamics} into a nominal nonlinear system and an uncertain parametric \ac{LTV} system, including a \emph{convex} over-bound of the linearization errors and parametric uncertainties.
Specifically, leveraging established \ac{SLS} tools for parametric \ac{LTV} systems~\cite{Chen2022RobustSynthesis}, we jointly optimize the linearization error bounds, parametric uncertainties bounds, nonlinear trajectory and error feedback.
This approach guarantees robust constraint satisfaction as outlined in Section~\ref{sec:sls}.
It leads to a joint optimization of the nominal nonlinear trajectory~\eqref{eq:ref_dynamics} and the affine time-varying error feedback~\eqref{eq:affine_feedback}. The resulting nonlinear optimization problem is detailed in Section~\ref{sec:rnoc_sls}.
Compared to related work~\cite{Leeman2023RobustSynthesis}, the proposed approach considers parametric uncertainties and achieves a convex over-bound of the uncertainties, reducing conservatism in handling uncertainties.
\subsection{Over-approximation of the nonlinear reachable set}
\label{sec:linearization}
The following outlines the re-formulation of the uncertain nonlinear system into an equivalent form as an uncertain \ac{LTV} system.
The linearization of the dynamics~\eqref{eq:nonlinear_dyn} around a nominal state and input $(z,v)$ is characterized by the Jacobian matrices
\begin{equation}
\begin{aligned}
&A(z, v, \theta) \defmath  \left.\frac{\partial  f}{\partial x}   \right|_{\substack{x=z\\u=v}},~B(z, v, \theta) \defmath \left.\frac{\partial  f}{\partial u}   \right|_{\substack{x=z\\u=v}},
\label{eq:jac}
\end{aligned}
\end{equation}
and the symmetric Hessian $H_i$ of the $i^\text{th}$ component of $f$
\begin{equation}
\label{eq:hess}
 H_i(\xi,\theta) = \left.\begin{bmatrix}
     \frac{\partial^2  f_i}{\partial x^2}  & \frac{\partial^2  f_i}{\partial x\partial u}\\
     * & \frac{\partial^2  f_i}{\partial u^2} 
 \end{bmatrix}\right|_{(x,u,\theta) = (\xi,\theta)},
\end{equation}
evaluated at some {(unknown)} point $\xi\in \R^{\nx + \nu}$, between $(z,v)$ and $(x,u)$, {i.e., $\xi\in\{(1-\alpha) (z,v) + \alpha (x,u)~|~0\le \alpha \le 1 \}$.}
We apply the Taylor series expansion to~\eqref{eq:nonlinear_dyn} and use the definitions~\eqref{eq:jac} to obtain
\begin{align}
    f(x ,u, \theta ) + w \stackrel{ \eqref{eq:dynamics}}{=}
    & \bar f(z , v) + f_\theta (z,v)\theta  \nonumber\\
    &+  A(z ,v, \theta  )( x  - z ) +B(z ,v , \theta )( u  - v )\nonumber\\
   & + {r(x ,u ,z ,v, \theta)}+ w,\label{eq:lin_hess}
\end{align}
where $ r(x ,u ,z ,v, \theta) \defmath (e\T H_1 e,\dots , e\T H_\nx e)$ is the Lagrange remainder accounting exactly for the approximation resulting from the linearization, and ${e} \defmath {(x-z, u-v)} =  (\Delta x, \Delta u) \in \R^{\nx + \nu}$ is the error.
This Lagrange remainder is overbounded, similarly as in~\cite{althoff2008reachability}, using the (worst-case) curvature $\mu \in \R^{\nx \times \nx}$ defined as
\begin{equation}\label{eq:def_mu}
    \mu \defmath \text{diag}(\mu_1,\dots, \mu_\nx), ~\mu_i\defmath \frac{1}{2} \max_{\substack{\xi \in \mathcal{C},  \theta \in \Theta,\\ \epsilon\in \B^{\nx + \nu}}} |\epsilon^\top H_i(\xi,\theta)\epsilon|,
\end{equation}
over-bounded within the constraint set $\mathcal{C}$.
The following Proposition establishes that, by using the maximization over all possible values of $\xi$ and $\theta$, the constants $\mu$ capture the worst-case linearization error.
\begin{proposition}
\label{prop:lin_err}
Given Assumptions \ref{assum:1} and \ref{assum:2}, the remainder in \eqref{eq:lin_hess} satisfies
\begin{equation}
    |r_i(x,u,z,v,\theta)|\le \|e\|_\infty^2 \mu_i,
\end{equation}
for any $(x,u) \in \mathcal{C}, (z,v) \in \mathcal{C}, \theta \in \Theta$.
\end{proposition}
\begin{proof}
{See~\cite{Leeman2023RobustSynthesis}.}
\end{proof}
The evaluation of the constants $\mu_i$ represents our proposed approach's only offline design requirement, distinguishing it from other prominent strategies, such as those presented in~\cite{Kohler2021ASystems,Sasfi2022RobustMetrics}, which necessitate the design of a (nonlinear) incrementally stabilizing controller. The $\mu_i$ reflects the level of nonlinearity present in the system and is used in over-approximating the disturbance set. 

We can compute an outer approximation of the reachable set of the nonlinear system~\eqref{eq:dynamics}, using the \ac{LTV} error system between the system dynamics~\eqref{eq:ref_dynamics} and~\eqref{eq:lin_hess}, i.e.,
\begin{equation}
\begin{aligned}
\Delta \X_{k+1} &= \X_{k+1} - \Z_{k+1} \\
&\stackrel{\eqref{eq:ref_dynamics},\eqref{eq:lin_hess}}{=} A(\Z_k, \V_k, \theta) \Delta\X_k + B(\Z_k, \V_k, \theta) \Delta\U_k\\ 
&~+ \underbrace{\fa(\Z_k, \V_k)\theta  + r(\X_k ,\U_k ,\Z_k ,\V_k, \theta) + \W_k}_{\mathdef \mathbf{d}_k}\\
\end{aligned}
\label{eq:LTV_error_nl}
\end{equation}
with 
\begin{equation}
\label{eq:def_ABtheta}
    \begin{aligned}
        A(\Z_k, \V_k, \theta ) &\defmath \bar A(\Z_k, \V_k) + A_\theta(\Z_k, \V_k)  \theta, \\
        B(\Z_k, \V_k, \theta ) &\defmath \bar B(\Z_k, \V_k) + B_\theta (\Z_k, \V_k)  \theta,
    \end{aligned}
\end{equation}
where $\bar A$ and $A_\theta$ are the Jacobians of $\fl$ and $f_\theta$, respectively, computed similarly as in~\eqref{eq:jac}.
Due to Proposition~\ref{prop:lin_err}, we have $r(x,u,z,v,\theta) \in \|e\|_\infty^2\mu \B^\nx$, and the lumped disturbance $\mathbf{d}_k$ in~\eqref{eq:LTV_error_nl} satisfies
\begin{equation}
\begin{aligned}
    \mathbf{d}_k&\in  \fa(\Z_k, \V_k)\Theta  \oplus \|\mathbf{e}_k\|_\infty^2 \mu \B^\nx \oplus \mathcal{W} \\
   &=  \fa(\Z_k, \V_k)\Theta  \oplus [E, \|\mathbf{e}_k\|_\infty^2 \mu] \B^{\nw + \nx}.
\end{aligned}
\label{eq:d_set}
\end{equation}
This set is expressed as a combination of three sets: ${ \fa(\Z_k, \V_k)\Theta }$, which is associated to the parametric uncertainty, $\|\mathbf{e}_k\|_\infty^2 \mu \B^\nx$, which reflects the linearization error, and $\mathcal{W}$, which captures the additive disturbance $\mathbf{w}_k$.
In Section~\ref{sec:rnoc_sls}, we will discuss a {jointly optimized} over-bounding of the uncertainties in~\eqref{eq:d_set}.

To optimize an affine error feedback~\eqref{eq:affine_feedback} while ensuring robust constraint satisfaction of the uncertain nonlinear system~\eqref{eq:nonlinear_dyn} based on the \ac{LTV} error dynamics~\eqref{eq:LTV_error_nl}, we study the robust optimal control problem for the special case of uncertain parametric \ac{LTV} systems in the following.

\subsection{Robust optimal control for parametric \ac{LTV} systems via \ac{SLS}}
\label{sec:sls}
We explore the system level parameterization of affine error feedback for uncertain \ac{LTV} systems, inspired by the technique proposed in~\cite{Chen2022RobustSynthesis}.
To that end, we consider a nominal \ac{LTV} system of the form
\begin{equation}
\tilde\Z_{k+1} = \tilde A_k \tilde\Z_k +\tilde B_k \tilde\V_k,~\tilde\Z_0 = \xz,
    \label{eq:nom_dyn}
\end{equation}
and an error \ac{LTV} system of a similar form as~\eqref{eq:LTV_error_nl},\begin{equation}
\begin{aligned}
\label{eq:uncertain_LTV}
   \Delta \tilde \X_{k+1} &= \tilde A_k \Delta \tilde \X_k + \tilde B_k \Delta\tilde  \U_k+ \underbrace{ \Delta \tilde A_k \Delta \tilde \X_k + \Delta \tilde B_k \Delta \tilde \U_k +  \d_k}_{\mathdef \E_k},
\end{aligned}
\end{equation}
where $\d_k\in  \tilde{\mathcal{D}}_k \defmath P_{\theta, k} \Theta \oplus P_{k}\B^\nx,$ with $P_{\theta, k}\in \R^{\nx\times \nt}$ and $P_k\in \R^{\nx \times \nwt}$, which is similar to the set~\eqref{eq:d_set}, and 
\begin{equation}
    \Delta \tilde A_k \in A_{\theta,k} \Theta,~ \Delta \tilde B_k \in B_{\theta,k} \Theta,\label{eq:def_Delta}
\end{equation}analogous to~\eqref{eq:def_ABtheta}.
The notation $\tilde{\cdot}$ differentiates~\eqref{eq:uncertain_LTV} from the approximated nonlinear system~\eqref{eq:LTV_error_nl}.
This parameterization will be used with the \ac{LTV} system~\eqref{eq:LTV_error_nl} and hence for parameterizing affine error feedback for the uncertain nonlinear system~\eqref{eq:nonlinear_dyn}. 

The system dynamics~\eqref{eq:uncertain_LTV} can be expressed in a compact form by stacking the quantities over the horizon, i.e.,
\begin{equation}
\label{eq:error_uncertain_LTV}
    \Delta \tilde \X = \mathbf{Z}\tilde{\mathbf{A}}\Delta \tilde \X + \mathbf{Z}\tilde{\mathbf{B}}\Delta \tilde \U + \underbrace{\mathbf{Z}(\tilde \Delta_A \Delta \tilde \X +\tilde \Delta_{B}\Delta \tilde \U)+ \d}_{=:\E},
\end{equation}
with {$\mathbf{Z}$, a block-downshift operator commonly used in \ac{SLS} (see, e.g., \cite{Leeman2023RobustSynthesis}), i.e.,
a matrix with identity matrices along its first block sub-diagonal and zeros elsewhere}, $\Delta \tilde \X \defmath \Delta \tilde \X_{1:T},~\Delta \tilde \U \defmath \Delta \tilde \U_{1:T},~\d \defmath \d_{0:T-1},$ and $\tilde{\mathbf{A}} \defmath \text{blkdiag}(\tilde A_1, \ldots, \tilde A_{T-1}, 0),~\tilde \Delta_{A} \defmath \text{blkdiag}( \Delta \tilde  A_1, \ldots,  \Delta \tilde  A_{T-1}, 0),$ and analogously for $\tilde{\mathbf{B}}$ and $\tilde \Delta_{B}$.

\begin{remark}Unlike previous works such as~\cite{Anderson2019,sieber2021system}, we employ the same notation as in~\cite{Leeman2023RobustSynthesis} which distinguishes $\Delta \tilde \X_0$ from $\Delta\tilde \X$, leading to a simpler parameterization.\end{remark}

As in~\cite{Leeman2023RobustSynthesis}, we introduce a causal error feedback
\begin{equation}
     \Delta \tilde \U = \mathbf{K}  \Delta \tilde \X,~\tilde \U_0 = \tilde \V_0,~\mathbf{K}\in \mathcal{L}^{T,\nu \times \nx}.
    \label{eq:affine_contr}
\end{equation}
This feedback allows us to derive the closed-loop error dynamics of~\eqref{eq:error_uncertain_LTV} as
 \begin{equation}
   \Delta \tilde \X = \mathbf{Z}( \tilde{\mathbf{{A}}} +  \tilde {\mathbf{B}}   {\mathbf{K}})  \Delta  \tilde \X+ \E,~  \tilde \U = \tilde \V+ {\mathbf{K}} \Delta  \tilde \X,~   \Delta\tilde  \X_0=0.
\label{eq:closed-loop_uncertain_LTV}
\end{equation}
Ensuring robust constraint satisfaction requires tight bounds on the lumped disturbance $\E$. However, obtaining such bounds can be challenging due to the dependence of $\E$ on the realized state $\Delta \tilde \X$ and input $\Delta \tilde \U$. To address this issue, we employ an over-approximation technique that {utilizes an additional decision variable}: a filtered disturbance set~\cite{Chen2022RobustSynthesis}. Therefore, we focus on the related \ac{LTV} system
\begin{equation}
   \Delta \tilde \X = \mathbf{Z}( \tilde{\mathbf{{A}}} +  \tilde {\mathbf{B}}   {\mathbf{K}})  \Delta  \tilde \X+ \mathbf \Sigma \tilde \W,~  \tilde \U = \tilde \V+ {\mathbf{K}} \Delta  \tilde \X,~   \Delta \tilde \X_0=0,
    \label{eq:filter_LTV}
\end{equation}
with the unit noise $\tilde \W_k \in \B^\nx$, and the filter
\begin{equation}
    \mathbf\Sigma  \defmath \text{blkdiag}(\mathbf\Sigma_0, \dots, \mathbf\Sigma_T),    \label{eq:Sigma}
\end{equation}
where $\mathbf{\Sigma}_k\defmath \text{diag}(\sigma_{k,1}, \dots, \sigma_{k,\nx})\in\R^{\nx\times\nx}$ are diagonal positive definite matrices.
The filter $\mathbf \Sigma$ is constrained to ensure that the reachable set of~\eqref{eq:filter_LTV} contains that of~\eqref{eq:closed-loop_uncertain_LTV}, both using the same controller. To guarantee this reachable set inclusion, the components of the disturbance $\mathbf{\Sigma}_k \tilde \W_k$ must have an amplitude at least as large as the lumped uncertainty $\E_k$ as defined in~\eqref{eq:uncertain_LTV}. The filter design is detailed in Proposition~\ref{prop:bound_dist}.

The closed-loop system~\eqref{eq:filter_LTV} is compactly written as
\begin{equation}
\begin{bmatrix} \Delta\tilde  \X\\ \Delta\tilde  \U \end{bmatrix} = \begin{bmatrix} (I - \mathbf{Z} \tilde{\mathbf{A}} - \mathbf{Z} \tilde{\mathbf{B}}  {\mathbf{K}} )^{-1} \\  {\mathbf{K}}(I - \mathbf{Z} \tilde{\mathbf{A}} - \mathbf{Z} \tilde{\mathbf{B}} {\mathbf{K}})^{-1} \end{bmatrix} \mathbf{\Sigma} \tilde{\W}= \vcentcolon \begin{bmatrix}  \Px \\  \Pu \end{bmatrix}  \tilde{\W},\\
\label{eq:sys_rep}
\end{equation}
with $\Px\in\mathcal{L}^{T,\nx \times \nx}, ~\Pu\in\mathcal{L}^{T,\nu \times \nx}$.
The following proposition presents a parameterization of the affine error feedback~\eqref{eq:affine_contr} as used in the \ac{LTV} system~\eqref{eq:filter_LTV}.
\begin{proposition}[Adapted from~\cite{Chen2022RobustSynthesis}]
\label{prop:slp}
The following two statements are equivalent
\begin{enumerate}    \item 
Let $\tilde \W\in { \B^{T\nx}}$ be an arbitrary disturbance sequence.
Any error state and input $\Delta\tilde \X$, $ \Delta \tilde\U$ satisfying~\eqref{eq:filter_LTV}, can also be represented as in~\eqref{eq:sys_rep} with some $\Px \in \mathcal{L}^{T,\nx \times \nx}$, $\Pu \in \mathcal{L}^{T,\nu \times \nx}$ and diagonal $\mathbf\Sigma$ as in~\eqref{eq:Sigma}, lying in the affine subspace
\begin{equation}
[I-\mathbf{Z} \tilde{\mathbf{A}}, ~- \mathbf{Z} \tilde{\mathbf{B}}]
\begin{bmatrix}
    \Px\\
    \Pu
\end{bmatrix}=\mathbf\Sigma.
\label{eq:slp}
\end{equation}
\item 
Let $\Px \in \mathcal{L}^{T,\nx \times \nx}$, $\Pu \in \mathcal{L}^{T,\nu \times \nx}$ and diagonal $\mathbf\Sigma$ be arbitrary matrices that satisfy~\eqref{eq:slp}. Then, the corresponding error state and input $\Delta \tilde \X$ and $\Delta \tilde \U$ computed using~\eqref{eq:sys_rep} satisfy the closed-loop dynamics in~\eqref{eq:filter_LTV} with the feedback gains ${\mathbf{K}} = \Pu  \Px^{-1}\in \mathcal{L}^{T,\nu\times \nx}$.
\end{enumerate}
\end{proposition}
\begin{proof}
    See~\cite{Chen2022RobustSynthesis}.
\end{proof}

For any nominal trajectories $\tilde \Z$, $\tilde \V$ satisfying \eqref{eq:nom_dyn}, and any error feedback $\Delta \tilde \U = \Pu  \Px^{-1} \Delta \tilde \X, ~\tilde \U_0 = \tilde \V_0$, with $\Px,\Pu$ satisfying~\eqref{eq:slp}, the closed-loop error on the states and inputs for the \ac{LTV} system~\eqref{eq:filter_LTV} is given by
\begin{equation}
     \sv{ \Delta \tilde \X_k}{ \Delta \tilde \U_k}= \P^{(k-1)} \tilde\W~\K,
    \label{eq:decomp}
\end{equation}
where $ \P^{(k)} \defmath  \sv{\Px^{(k)}}{\Pu^{(k)}}$.
Hence, the disturbance reachable sets of the system~\eqref{eq:filter_LTV} can be characterized as
\begin{equation}
\begin{aligned}
    \tilde \X_k &\in \mathcal{R}_\x(\tilde \Z_k,\Px^{(k-1)}) \defmath \{\tilde \Z_k\}\oplus \Px^{(k-1)} \B^{T \nx},\\
    \tilde \U_k &\in \mathcal{R}_\u(\tilde \V_k,\Pu^{(k-1)}) \defmath \{\tilde \V_k\}\oplus \Pu^{(k-1)} \B^{T \nx}.
\end{aligned}
\label{eq:reach}
\end{equation}

We now consider the conditions that the filter $\mathbf \Sigma$ must satisfy for the reachable set of~\eqref{eq:filter_LTV} to include the reachable set of~\eqref{eq:closed-loop_uncertain_LTV} when both systems are in closed-loop with the same controller $\mathbf{K}$.
This condition allows us to construct a controller for~\eqref{eq:uncertain_LTV} based on the system~\eqref{eq:filter_LTV}.
\begin{proposition}
\label{prop:bound_dist}
Let $\mathbf{K}=\Pu  \Px^{-1}\in \mathcal{L}^{T,\nu \times \nx}$ be the gains of an affine error feedback~\eqref{eq:affine_contr} parameterized as in Proposition~\ref{prop:slp} and let the nominal trajectories $\Z, \V$ satisfy~\eqref{eq:nom_dyn}. For any realization of $\theta\in\Theta$, the reachable set of the uncertain \ac{LTV} system~\eqref{eq:closed-loop_uncertain_LTV} is a subset of the reachable set of the filter-based \ac{LTV} system~\eqref{eq:filter_LTV}, starting from the same initial condition, if the inequality holds
\begin{equation}
    \onenorm{[ A_{\theta,k} \theta \Px^{(k-1)} + B_{\theta,k} \theta \Pu^{(k-1)}  ,  P_{\theta, k} \theta,   P_k ] } \le \sigma_{k,i},
\label{eq:filer_constr}
\end{equation}
for each time step $k\in \mathbb{N}_T$, each vertex $ \theta \in \vertx{\Theta}$, and for all $ i\in \mathbb{N}_{\nx}$, where $\e_i$ is the $i^\text{th}$ row of the identity matrix $I_{\nx}$, and $\vertx{\Theta}$ is the set of vertices of $\Theta$.
\end{proposition}
\begin{proof}
The proof is based on the results from~\cite{Chen2022RobustSynthesis}. As the initial condition $\Delta \tilde \X_0$ and the matrices $\tilde{\mathbf A}$, $\tilde {\mathbf B}$ and $\mathbf{K}$ are identical for the \ac{LTV} systems~\eqref{eq:closed-loop_uncertain_LTV} and~\eqref{eq:filter_LTV}, it suffices to show that the disturbance $\E$ in~\eqref{eq:closed-loop_uncertain_LTV} is contained in the filtered ball~$\mathbf \Sigma\B ^{T\nx}$, for all $\Delta \X$, $\Delta \U$, $\d$, i.e.,
\begin{equation}
   \E_k \stackrel{\eqref{eq:uncertain_LTV}}{=} \Delta \tilde A _k\Delta\tilde  \X_k + \Delta \tilde{B}_k\Delta  \tilde \U_k+ \d_k \in \mathbf \Sigma \B ^{T \nx}.
\end{equation}
Hence, in element-wise notation, the following inequality
\begin{equation}
| \e_i^\top (\Delta\tilde  A_k \Delta\tilde  \X_k + \Delta\tilde  B_k \Delta\tilde  \U_k + \d_k)| \stackrel{\eqref{eq:Sigma}}{\le} \sigma_{k,i},
\label{eq:prop_eq_proof_1}
\end{equation}
has to be satisfied for any realization of $\d_k \in \tilde{\mathcal{D}}_k$, any $\theta\in \Theta$ and for each time step $k\in \mathbb{N}_T$.

After replacing $\Delta \tilde \X_k$, $\Delta\tilde \U_k$ using~\eqref{eq:decomp}, and $\Delta\tilde  A_k$, $\Delta\tilde B_k$ by their definitions~\eqref{eq:def_Delta}, we use the triangle inequality and take the max of its left-hand side. This ensures that the inequality is satisfied for any $\tilde \W_k\in \B^\nx$ and $\d_k \in \tilde{\mathcal{D}}_k$, i.e.,
\begin{equation}
\begin{aligned}
\max_{\tilde{\W}\in {\B^{T\nx}}} |\e_i^\top (A_{\theta,k} \theta \Px^{(k-1)} &+ B_{\theta,k} \theta \Pu^{(k-1)}) \tilde \W| \\
&\quad + \max_{\d_k \in \tilde{\mathcal{D}}_k} | \e_i^\top \d_k | \le \sigma_{k,i}.    
\end{aligned}
\label{eq:proof_bounds_3}
\end{equation}

Then, we leverage the convexity of~\eqref{eq:proof_bounds_3} in $\theta$ and that the maximum of a convex function over a convex polytope is achieved at the vertices. We finally obtain~\eqref{eq:filer_constr} using the definition of the 1-norm.
\end{proof}

\begin{remark}
The results presented in this paper can be extended to a general invertible filter $\mathbf{\Sigma} \in \mathcal{L}^{T,\nx\times \nx}$~\cite{Chen2022RobustSynthesis}.
\end{remark}

By applying Propositions~\ref{prop:slp} and~\ref{prop:bound_dist}, the reachable sets~\eqref{eq:reach} provide an over-approximation of those resulting from~\eqref{eq:closed-loop_uncertain_LTV}, allowing us to derive sufficient conditions for robust constraint satisfaction, as stated in the following proposition.

\begin{proposition}
\label{prop:lin}
If there exist matrices $\Px$, $\Pu$, $\mathbf{\Sigma}$ and a nominal trajectory satisfying~\eqref{eq:slp},~\eqref{eq:nom_dyn},~\eqref{eq:filer_constr}, and
\begin{equation}
c_i\T \sv{\tilde\Z_k}{ \tilde\V_k}  + b_i+ \| c_i\T \P^{(k-1)}\|_1 \le 0~\K~\I,
    \label{eq:nom_noise}
\end{equation}
then, it is guaranteed that, for any $\d\in  \tilde{\mathcal{D}}^{0:T-1}$,
\begin{equation}
    c_i\T \sv{\tilde \X_k}{ \tilde\U_k}  + b_i \le 0~\K~\I,
 \label{eq:lin_const}
\end{equation}
with $\sv{ \tilde \X_k}{\tilde  \U_k}$ according to~\eqref{eq:closed-loop_uncertain_LTV} and~\eqref{eq:nom_dyn}, with ${\mathbf{K}} = \Pu  \Px^{-1}$.
\end{proposition}
\begin{proof}
    See~\cite{Chen2022RobustSynthesis}.
\end{proof}
To derive an affine error feedback that robustly satisfies constraints for uncertain \ac{LTV} dynamics such as~\eqref{eq:uncertain_LTV}, we leverage Propositions~\ref{prop:slp},~\ref{prop:bound_dist}, and~\ref{prop:lin}.
Specifically, we parameterize the controller used in the \ac{LTV} system~\eqref{eq:filter_LTV} using Proposition~\ref{prop:slp}, and constrain its filter using Proposition~\ref{prop:bound_dist} to ensure that the reachable set of the \ac{LTV} system~\eqref{eq:filter_LTV} includes that of the uncertain \ac{LTV} system~\eqref{eq:uncertain_LTV}.
Furthermore, we use Proposition~\ref{prop:lin} to derive {sufficient} conditions for robust constraint satisfaction for the uncertain \ac{LTV} system~\eqref{eq:uncertain_LTV}. By combining these results, we obtain the following optimization problem that solves Problem~\eqref{eq:prob_form} for an affine feedback policy in the particular case of uncertain \ac{LTV} dynamics:
\begin{subequations}
    \begin{align}
    \min_{\substack{ \tilde \Z,  \tilde \V_0, \tilde \V,\\ \P, \mathbf\Sigma}}  &  J_T(\xz,  \Z, \V,\P),\\
    \text{s.t.}  & \left[ I - \mathbf{Z}\tilde{\mathbf{A}},~ -\mathbf{Z}\tilde{\mathbf{B}}\right]\begin{bmatrix}\Px\\\Pu \end{bmatrix}= \mathbf\Sigma,\\
    &\Km:\nonumber\\
    &  \tilde \Z_{k+1} =  \tilde A_k \tilde  \Z_k +  \tilde B_k  \tilde \V_k,~\tilde \Z_0 = \xz,\\
    & \forall i\in \mathbb{N}_{n_\x}:\nonumber\\
    &     \onenorm{[\tilde \Delta_k \theta \P^{(k-1)},  P_{\theta, k} \theta,   P_k ] } \le \sigma_{i,k}~ \forall \theta \in \vertx{\Theta},\label{eq:sls_param_overbound}\\
    &  \| c_i\T \P^{(k-1)} \|_1+c_i\T \sv{\tilde \Z_k}{\tilde \V_k} + b_i  \le 0~\I,
    \end{align}
    \label{eq:sls_param}
\end{subequations}
where we use the shorthand notation $\tilde \Delta_k \defmath [ {A}_{\theta, k},  {B}_{\theta, k}]$.
\begin{remark}
If the set of uncertain parameters $\Theta$ is reduced to a singleton, which implies the absence of parametric uncertainties,~\eqref{eq:sls_param} can be simplified to a classical \ac{SLS} problem, as shown, e.g., in~\cite[Eq.(15)]{leeman2022predictive} or~\cite{sieber2021system}.
\end{remark}
\subsection{Robust nonlinear finite-horizon optimal control problem}
\label{sec:rnoc_sls}
We can use the system level parameterization of the affine error feedback for uncertain \ac{LTV} systems in~\eqref{eq:error_uncertain_LTV} on the similar uncertain \ac{LTV} system~\eqref{eq:LTV_error_nl}.
The uncertain \ac{LTV} system~\eqref{eq:LTV_error_nl} is accounting for the linearization errors of the nonlinear system as described in Proposition~\ref{prop:lin_err}. 
With this in mind, we can now present the central result of this paper.
We utilize Proposition~\ref{prop:bound_dist} to derive a convex over-approximation of the lumped uncertainties that arise from the additive noise, the multiplicative disturbance, and the linearization errors. 
To construct the robust optimal control problem~\eqref{eq:sls}, we also leverage the system level parameterization of affine error feedback (Proposition~\ref{prop:slp}) and the sufficient condition for robust constraint satisfaction (Proposition~\ref{prop:lin}).
The solution to the robust optimal control problem~\eqref{eq:sls} yields an error feedback and a nonlinear nominal trajectory that, when applied to the uncertain nonlinear system, guarantee robust constraint satisfaction, as shown in Theorem~\ref{thm:1}.
\begin{figure*}[!t]
\begin{subequations}
    \begin{align}
    \min_{ \Px, \Pu,\mathbf\Sigma,\Z,\V_0,\V,\tube} \quad &  J_T(\xz, \Z,\V, \P), \\
    \text{s.t.}\quad   &\left[ I - \mathbf{Z}\bar{\mathbf{A}}(\Z,\V),~ - \mathbf{Z}\bar{\mathbf{B}}(\Z,\V)  \right]  \begin{bmatrix}
     \Px\\
     \Pu\\
    \end{bmatrix} = \mathbf\Sigma,\label{eq:sls_slp}\\
    & \Z_{k+1} = \fl(\Z_k, \V_k)~\Km,~\Z_0= \xz,\label{eq:nonlinear_sls_nom}\\
    & \|\e_i^\top[ \Delta_k \theta \P^{(k-1)} ,\fa(\Z_k, \V_k)\theta, E, \tube_k^2 \mu]\|_1 \le \sigma_{i,k}~\Km ~\forall i\in \mathbb{N}_{\nx}~ \forall \theta \in \vertx{\Theta},\label{eq:bound_constr} \\
    & \| c_i\T  \P^{(k-1)}\|_1 + c_i\T \sv{\Z_k}{\V_k} + b_i \le 0~\Km~\I,\label{eq:cons_nonlinear_SLS}\\
    &\|\P^{(k-1)}\|_\infty\le \tube_k~\Km.
    \end{align}
    \hrulefill
    \vspace*{-15pt}
    \label{eq:sls}
\end{subequations}
\end{figure*}
\begin{theorem}
Given Assumptions \ref{assum:0}, \ref{assum:1}, \ref{assum:2} and \ref{assum:3}, suppose optimization problem \eqref{eq:sls} is feasible for some $\bar x$. Then, the affine error feedback $\U = \V^\star + \mathbf{K}^\star (\X - \Z^\star),~\mathbf{K}^\star = \Pu^\star\Px^{\star -1},~\U_0 = \V_0^\star$ obtained from~\eqref{eq:sls} provides a feasible solution to Problem~\eqref{eq:prob_form}, i.e., the closed-loop trajectories of system~\eqref{eq:dynamics} under this error feedback
robustly satisfy the constraints~\eqref{eq:prob_form_cons}.
\label{thm:1}
\end{theorem}
\begin{proof}
First, the uncertain nonlinear system~\eqref{eq:dynamics} is conservatively reformulated as the uncertain \ac{LTV}~\eqref{eq:LTV_error_nl} using a Taylor series approximation with respect to the nominal trajectory~\eqref{eq:nonlinear_sls_nom}.
Then, we apply Proposition~\ref{prop:slp} to an equivalent filter-based error system
\begin{equation}
\begin{aligned}
 \Delta \X_0 &=0,\\
\Delta \X_{k+1} &= \bar A(\Z_k, \V_k) \Delta\X_k + \bar B(\Z_k, \V_k) \Delta\U_k + \mathbf{\Sigma}_k \tilde \W_k,
\label{eq:proxy_LTV}
\end{aligned}
\end{equation}
constructed based on the \ac{LTV} error~\eqref{eq:LTV_error_nl}. The constraint~\eqref{eq:sls_slp} implies that the trajectories of~\eqref{eq:proxy_LTV} satisfy
\begin{equation}
    \begin{bmatrix}
        \X - \Z\\
        \U - \V\\
    \end{bmatrix} =
    \begin{bmatrix}
        \Px\\
        \Pu
    \end{bmatrix}
    \tilde{\W}.
    \label{eq:error}
\end{equation}
Thus, we apply Proposition~\ref{prop:bound_dist} to guarantee that the reachable sets of~\eqref{eq:error} include that of~\eqref{eq:dynamics}. Then, the definition of the disturbance set~\eqref{eq:d_set} leads to the following inequality$~\Km, ~\forall \theta \in \Theta$
\begin{equation}
    \|\e_i^\top[ \Delta_k \theta \P^{(k-1)},\fa(\Z_k, \V_k)\theta, E, \|\mathbf{e}_k\|_\infty^2 \mu]\|_1 \le \sigma_{i,k},
\label{eq:proof_step_over_bound}
\end{equation}
with $ \Delta_k \defmath [ A_\theta(\Z_k, \V_k) ,  B_\theta(\Z_k, \V_k) ]$.
Additionally, the constraint~\eqref{eq:cons_nonlinear_SLS} guarantees that the constraints are robustly satisfied as per Proposition~\ref{prop:lin}. Finally, using~\eqref{eq:error}, the error $\mathbf{e}_k$ in~\eqref{eq:proof_step_over_bound} can be over-approximated with
\begin{equation}
    \|\mathbf{e}_k\|_\infty \le \|\P^{(k-1)}\|_\infty\le \tube_k,
\end{equation}
where $\tube_k$ is a jointly optimized auxiliary variable.
\end{proof}
The optimization problem~\eqref{eq:sls} jointly optimizes the nonlinear trajectory $(\Z, \V)$, error feedback $(\Px, \Pu)$, convex over-bounds $\bm{\Sigma}$, and linearization error bounds through $\bm{\tau}$. It is a new formulation compared to the one proposed in~\cite{Leeman2023RobustSynthesis}, even when the set of uncertain parameters $\Theta$ reduces to a singleton, indicating that there are no parametric uncertainties. This key difference is due to using a filter that lumps the effects of uncertainties, as shown in Proposition~\ref{prop:bound_dist}.
In particular, the reachable set of the nonlinear system~\eqref{eq:nonlinear_dyn} in closed-loop with the affine error feedback computed as in Theorem~\ref{thm:1} satisfy
\begin{equation}
\begin{aligned}
     \X_k &\in \mathcal{R}_\x( \Z_k^\star,\Px^{(k-1)\star }) \defmath \{ \Z_k^\star\}\oplus \Px^{(k-1)\star } \B^{T \nx},\\
     \U_k &\in \mathcal{R}_\u( \V_k^\star,\Pu^{(k-1)\star }) \defmath \{ \V_k^\star\}\oplus \Pu^{(k-1)\star } \B^{T \nx},
\end{aligned}
\label{eq:reach_NL}
\end{equation}
resulting in fewer nonconvex constraints compared to the formulation in~\cite{Leeman2023RobustSynthesis}.
However, as opposed to a linear~\ac{SLS}~\cite{Anderson2019}, the constraint~\eqref{eq:sls_slp} is in general not convex, since the nonlinear trajectory $(\Z,\V)$ is jointly optimized.
In the next section, we discuss two performance enhancement techniques enabled by the proposed formulation, which would not be possible with the previously proposed approach~\cite{Leeman2023RobustSynthesis}.

\section{Remarks on performance enhancement}
In this section, we discuss how to impose robust performance guarantees with convex constraints and perform a posteriori estimation of the model mismatch.

\subsection{Robust performance guarantees}
The parameterization of the error feedback presented in Problem~\eqref{eq:sls} can be augmented to obtain robust performance guarantees.
In particular, we study the performance of the resulting closed-loop nonlinear error dynamics~\eqref{eq:LTV_error_nl} and consider the controlled output signal, denoted by $\Y \defmath \mathbf{C}\Delta\X + \mathbf{D}\Delta\U$, where $\mathbf{C}\in\mathcal{L}^{T,\ny \times \nx}$ and $\mathbf{D}\in  \mathcal{L}^{T,\ny \times \nu}$ are typically diagonal and user-defined.
The goal is to satisfy the performance requirement given by
\begin{equation}
    \|\Y\|_\infty\le \gamma~\forall \W\in \mathcal{W}^T~\forall\theta\in\Theta,\label{eq:perf}
\end{equation}
where $\gamma\in\R$ is a user-defined performance index (see, e.g.,~\cite{magni2003robust, matni2020robust} for related definitions).

\begin{corollary}There exists a causal error feedback parameterized by~\eqref{eq:sls_slp} such that the robust performance constraint~\eqref{eq:perf} is satisfied, if the following inequality holds
\begin{equation}
    \|\mathbf{C}\Px + \mathbf{D} \Pu\|_\infty\le \gamma.
    \label{eq:gamma}
\end{equation}
\end{corollary}
The constraint~\eqref{eq:gamma} is convex in the decision variables, which is not possible with the parameterization used in~\cite{Leeman2023RobustSynthesis}.

\subsection{Learning the model mismatch}
\label{sec:4b}
Since the dynamics~\eqref{eq:dynamics} are affine in the parameter $\theta$, the realized trajectories $\X$, $\U$ can be leveraged to refine the set of all consistent parameter values $\Theta_p\subset \Theta$, based on set-membership estimation~\cite{milanese1991sm}, hence learning the model mismatch $\fa(z, v)\theta$.

\section{Numerical Application Example}
\begin{figure*}[ht!]
    \centering
    \includegraphics[width = \textwidth]{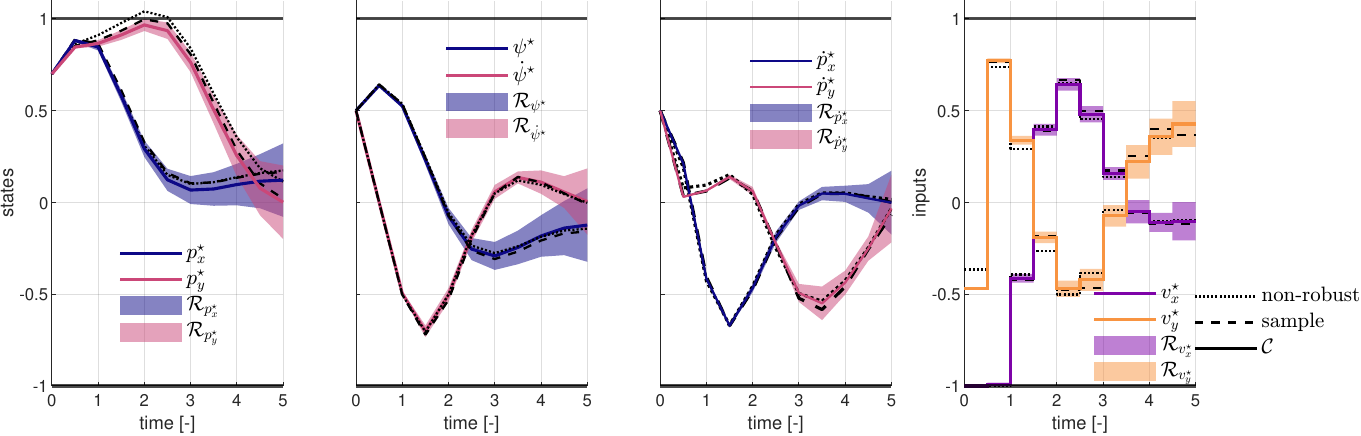}
    \caption{Solution of the \ac{SLS}-based robust nonlinear optimal control problem~\eqref{eq:sls} for a post-capture satellite stabilization simulation example. 
    The figure shows three pairs of states and one pair of inputs. For each pair, both nonlinear nominal trajectories (solid) and corresponding reachable sets (shaded areas) according to~\eqref{eq:reach_NL} are {jointly optimized} to remain within the constraints (bold black). This guarantees that any ``sample'' trajectory (dashed) satisfies the constraints. The ``non-robust'' (dotted) trajectory illustrates a trajectory resulting from an optimal control problem without robustness guarantees.
    }
    \label{fig:fig1}
\end{figure*}
\begin{figure*}[ht!]
    \centering
    \includegraphics[width = \textwidth]{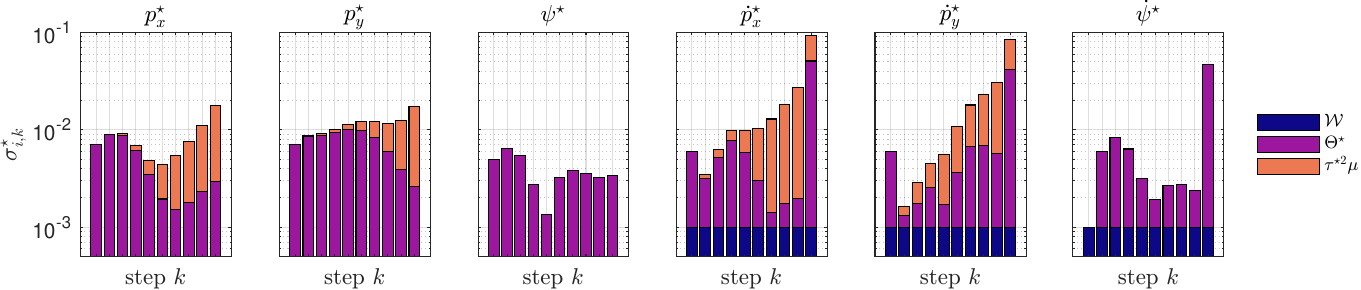}
    \caption{Optimal components of the filter $\mathbf{\Sigma}^\star$ decomposed according to~\eqref{eq:bound_constr} for the problem considered. The filter is designed to lump, at each time step, the effect of the uncertainties coming from the additive noise ($\mathcal{W}$), the parametric uncertainties ($\Theta^\star$), and the linearization errors (${\tube^\star}^2 \mu$).}
    \label{fig:fig2}
\end{figure*}
The considered example is motivated by the problem of removing space debris from Earth's orbit, which is approached by designing a chaser spacecraft to capture targets with largely unknown mass, inertia, position, or velocity. 
The simulation example in the following focuses on the challenge of achieving steady-state operation with significant uncertainty in the inertia.
\subsection{Satellite post-capture stabilization}
We consider a planar rigid-body model to study the problem of post-capture stabilization, as proposed in~\cite{virgili2019Simultaneous,zagaris2018reachability}.  The equations of motion are given by
\begin{equation}
    \ddot{{p}}(t) = R(\psi(t))\frac{(v_x(t), v_y(t))}{m} ,~\ddot \psi(t) = \frac{l\cdot v_x(t)}{j},
\label{eq:nonlin_dyn_rot}
\end{equation}
where $ R(\psi(t))\in\R^{2\times 2}$ is the rotation matrix, ${p} \defmath (p_x, p_y)\in\R^2$ is the relative distance to the target orbit, $m\in\R$ is the mass, $\psi\in\R$ is the relative angle, $j\in\R$ is the uncertain inertia, and $l$ is the moment arm. The state $z \defmath (p_x, p_y, \psi, \dot p_x, \dot p_y, \dot \psi ) \in \R^6$ and the input $v \defmath (v_x, v_y)\in\R^2$ are rendered dimensionless for numerical stability purposes. The dynamics are discretized using a forward-Euler integrator with a step size of 0.5 time units and 10 inner steps. We also impose constraints on the states and inputs, $  - 1 \le z_i \le 1~\forall i\in \mathbb{N}_6$ and $-1 \le v_i\le 1~\forall i\in \mathbb{N}_2$, as well as a robust performance constraint $ \|\Delta \X \|_\infty\le \gamma_\text{max},$ with $\gamma_\text{max} =0.2$.
A small bounded disturbance is applied to the system, described by $E = 10^{-3} [ 0_{3, 3},~ {I}_3]^\top \in \mathbb{R}^{6\times n_\mathrm{w}}$ with $n_\mathrm{w}=3$, according to~\eqref{eq:disturbance_E}. We consider uncertain inertia, such that $1/j \in \{1+ \delta|~ |\delta| \le 0.01 \}$.
We use a nominal cost function
 $J_T(\xz, \Z,\V)=\sum_{k=0}^{T-1}\ell(\Z_k,\V_k)+\ell_f(\Z_T)$, with the stage cost $ \ell(\Z_k,\V_k) \defmath \Z_k \T {Q}\Z_k + \V_k \T {R} \V_k$, and the terminal cost $\ell_f(\Z_T) \defmath \Z_T \T {Q}_f \Z_T,$ with ${Q} = {I}_6 $, ${R} = {I}_3$, $ Q_f = 10 Q$ and the horizon is $T = 10$. 
As common in numerical optimization, we include an additional regularization term in the cost function for numerical stability, such as $J_T + \lambda \bm{\beta}\T\bm{\beta}$, where $\lambda = 10^{-6}$ and $\bm{\beta}$ represents a vector collecting all the decision variables.
We approximate the constant $\mu \approx \text{diag}(0.68, 0.66, 0, 1.98, 1.95,0) \in \mathbb{R}^{6 \times 6}$ from Equation~\eqref{eq:def_mu} for the nonlinear dynamics~\eqref{eq:nonlin_dyn_rot} using a Monte-Carlo simulation.
The initial condition used is $\xz = (0.7,0.7,0.5,0.5,0.5,0.5)$.
\subsection{Results and discussion}
The \ac{NLP}~\eqref{eq:sls} is solved using the solver IPOPT~\cite{wachter2006implementation}, with its default settings, formulated with CasADi~\cite{Andersson2019}\footnote{
An open-source implementation is available at \url{https://gitlab.ethz.ch/ics/nonlinear-parametric-SLS}, doi: \url{https://doi.org/10.3929/ethz-b-000629589}.}.
For the problem considered, Fig.~\ref{fig:fig1} shows the solution of the \ac{NLP}~\eqref{eq:sls}, where the states and inputs are grouped in pairs. For each of the four pairs, we see the nominal trajectories (solid) and the corresponding reachable sets (shaded areas) computed according to~\eqref{eq:reach_NL}.
An illustrative disturbance sequence has been applied. As the proposed design guarantees, the resulting ``sample'' trajectory (dashed) remains within the reachable sets and, hence, within the constraint set (bold black).
The flexible error feedback parameterization allows the tubes to change size differently in each direction to meet the constraints, demonstrating the method's flexibility.

Moreover, Fig.~\ref{fig:fig1} compares our method with its nominal counterpart, which does not optimize error feedback ($\Px = 0, \Pu = 0$) and neglects disturbances ($\mathbf\Sigma =0$). The corresponding ``non-robust'' (dotted) trajectory with disturbances violates the constraints. The violation amplitude could be more pronounced when the constraints are active for extended durations, see e.g.,~\cite{Leeman2023RobustSynthesis}.

To highlight the importance of addressing parametric uncertainties tightly, we compare our method with an offline-overbounded counterpart similar to~\cite{Bujarbaruah2021AUncertainty}, where we over-bound offline the effect of the parametric uncertainties while still optimizing the over-bounding of the linearization error.
In this offline-overbounded approach, we introduce the disturbance set $\alpha \B^\nx$, where $\alpha$ is fixed offline and, for $\alpha = 10^{-2}$, accounts for the combined effect of the additive disturbance and uncertain parameter $\theta f_\theta (\X_k, \U_k)$, i.e., $ \theta f_\theta (x,u) \in 10^{-2} \B^\nx~\forall(x,u)\in\mathcal{C}~{\forall\theta \in \Theta}$.
Table~\ref{table:comparison} shows the optimal cost for different values of $\alpha$. 
The proposed method outperforms the offline-overbounded approach by being able to handle much larger disturbances:
the offline-overbounded approach is infeasible for $\alpha = 0.6 \cdot 10^{-2}$, i.e., $60\%$ of the disturbance amplitude that our method is feasible for. Additionally, the proposed method incurs only a marginal cost increase cost compared to the case that ignores parametric uncertainties, i.e., $\alpha =0$.

\begin{table}[htbp]
  \centering
  \caption{Performance comparison between the proposed method and the offline-overbounded approach.}
  \label{table:comparison}
  \begin{tabular}{lccc}
    \toprule
    & Uncert. set& $J_T (\bar x, \Z^\star, \V^\star)$ & $\alpha$ \\
    \midrule
    \textbf{Our method~\eqref{eq:sls}} &$ \Theta f_\theta(x,u)$ & 18.50& -- \\
    \midrule
    Eq.~\eqref{eq:sls}, &$\alpha \B^\nx$ & infeasible &$0.6 \cdot 10^{-2}$\\
    \cline{3-4}
     ~ with offline- & & 18.62& $0.5 \cdot 10^{-2}$ \\
         \cline{3-4}
   ~ overbound&  &18.33& $0.25 \cdot 10^{-2}$ \\
       \cline{3-4}
    &  &18.32 &0\\
    \bottomrule
    \vspace*{-25pt}
  \end{tabular}
\end{table}
Fig.~\ref{fig:fig2} shows the optimal values of the filter $\mathbf \Sigma^\star$ for the considered problem, decomposed according to~\eqref{eq:bound_constr}. 
For each time step and each state, the effect of the lumped disturbance is decomposed into three parts: the effect of the parametric uncertainties (purple), the effect of the linearization errors (orange), and the effect of the additive disturbance $\mathcal{W}$ (blue).
Since the dynamics are linear in the rotation, the linearization errors are zeros in the direction of $\psi$ and $\dot \psi$. The size of the purple region ($\Theta^\star$) highlights the importance of learning the model mismatch as in Section~\ref{sec:4b}. Again, this dynamic over-bounding demonstrates the flexibility of the method.

Note that the linearization of system~\eqref{eq:nonlin_dyn_rot} at the origin is not controllable, which prevents the use of classic linear control methods for comparison.
\section{Conclusion}
This paper has proposed a novel approach to solve finite-horizon constrained robust optimal control problems for nonlinear systems with affine parametric uncertainties and additive disturbances.
Our method simultaneously optimizes a nominal \emph{nonlinear} trajectory, an affine error feedback policy, and \emph{convex} uncertainty bounds and guarantees robust constraint satisfaction. The convex bounds enable deriving convex robust performance guarantees.
We demonstrated the effectiveness of our method in the simulation example of a post-capture stabilization of a satellite with state and input constraints and performance guarantees. Our results illustrate the performance and low conservatism of our approach.
For future research, we intend to work towards a corresponding recursive feasibility and stability analysis for a receding horizon implementation, in line with typical \ac{RMPC} practices. The flexibility of the tubes presents a particular challenge, as standard methods cannot be readily applied. {For practical usability, we intend to develop a tailored routine to solve efficiently the \ac{NLP}, as current available methods for \ac{MPC} cannot be applied.}
\bibliographystyle{IEEEtran}
\bibliography{IEEEabrv,references}
\balance
\end{document}

%% file: notes.tex
\newboolean{shownotes}
\setboolean{shownotes}{true}
\newcounter{noteCounter}

\definecolor{BgYellow}{HTML}{FFF59C}
\definecolor{FrameYellow}{HTML}{F7A600}
\definecolor{BgPink}{HTML}{EF6FA7}
\definecolor{FramePink}{HTML}{E5446E}
\definecolor{BgGreen}{HTML}{C7D92D}
\definecolor{FrameGreen}{HTML}{89B23B}
\definecolor{BgBlue}{HTML}{45BEE9}
\definecolor{FrameBlue}{HTML}{31A8C9}
\definecolor{BgWhite}{HTML}{D8D8D8}
\definecolor{FrameWhite}{HTML}{7F7F7F}
\definecolor{BgBrown}{HTML}{8E7A45}
\definecolor{FrameBrown}{HTML}{6B5B32}

\usepackage{contour}

\newtcolorbox{YStkyNote}[1][]{%
    enhanced,
    before skip=2mm,after skip=2mm, 
    width=0.45\textwidth, 
    boxrule=0.2mm,
    colback=BgYellow, colframe=FrameYellow, 
    attach boxed title to top left={xshift=0cm,yshift*=0mm-\tcboxedtitleheight},
    varwidth boxed title*=-3cm,
    boxed title style={frame code={%
        \path[left color=FrameYellow,right color=FrameYellow,
        middle color=FrameYellow]
        ([xshift=-0mm]frame.north west) -- ([xshift=0mm]frame.north east)
        [rounded corners=0mm]-- ([xshift=0mm,yshift=0mm]frame.north east)
        -- (frame.south east) -- (frame.south west)
        -- ([xshift=0mm,yshift=0mm]frame.north west)
        [sharp corners]-- cycle;
        },interior engine=empty,
    },
    sharp corners,rounded corners=southeast,arc is angular,arc=3mm,
    underlay={%
        \path[fill=BgYellow!80!black] ([yshift=3mm]interior.south east)--++(-0.4,-0.1)--++(0.1,-0.2);
        \path[draw=FrameYellow,shorten <=-0.05mm,shorten >=-0.05mm,color=FrameYellow] ([yshift=3mm]interior.south east)--++(-0.4,-0.1)--++(0.1,-0.2);
        },
    drop fuzzy shadow, 
    fonttitle=\bfseries, 
    title={#1}
}

\newcommand{\unnumberedfootnote}[1]{%
  \begingroup
  \renewcommand{\thefootnote}{}%
  \footnotetext{#1}%
  \addtocounter{footnote}{-1}%
  \endgroup
}

\newcommand{\note}[1]{%
\ifthenelse{\boolean{shownotes}}{%
\stepcounter{noteCounter}%
\hyperlink{note:\thenoteCounter}{\(^{\thenoteCounter}\)}%
\unnumberedfootnote{%
  \hypertarget{note:\thenoteCounter}{}%
  \begin{YStkyNote}
    \(^{\thenoteCounter}\)%
    #1%
  \end{YStkyNote}%
}%
}{}}

\newcommand{\dashed}[1]{%
  \ifthenelse{\boolean{shownotes}}{\st{#1}}{}%
}

\newcommand{\shortlong}[2]{%
  \ifthenelse{\boolean{shownotes}}{\blue{#2}}{#1}%
}